\documentclass[11pt,a4paper,USenglish]{article}

\usepackage[disable]{todonotes}

\usepackage[pagebackref]{hyperref}
\usepackage{amsthm}

\usepackage{fullpage}

\usepackage{mathtools}
\usepackage{amssymb}
\usepackage{comment}
\usepackage{amsfonts}

\usepackage{paralist}
\usepackage{tabularx}
\usepackage{float}

\usepackage{multicol}

\usepackage{tikz}
\usepackage{tikzscale}

\usetikzlibrary{calc}

\newtheorem{theorem}{Theorem}

\newtheorem{corollary}{Corollary}

\newtheorem{lemma}{Lemma}

\usepackage[noend]{algpseudocode}
\usepackage{caption}
\usepackage{subcaption}

\usepackage{enumerate}
\usepackage[capitalize]{cleveref}

\usepackage[numbers,sort]{natbib}

\usepackage{myprobdescr}

\newcommand{\decprob}[3]{%
    \begin{center}%
        \begin{minipage}{0.9\linewidth}%
            \defDecprob{#1}{#2}{#3}
        \end{minipage}%
    \end{center}%
}

\newcommand{\CI}{\textsc{Closeness Improvement}}
\newcommand{\BI}{\textsc{Betweenness Improvement}}
\newcommand{\DS}{\textsc{Dominating Set}}
\newcommand{\YI}{\textsc{Yes}-instance}
\newcommand{\NI}{\textsc{No}-instance}
\newcommand{\DBI}{\textsc{Directed Betweenness Improvement}}
\newcommand{\SC}{\textsc{Set Cover}}

\newcommand{\C}{\ensuremath{\mathcal{C}}}
\renewcommand{\S}{\ensuremath{\mathcal{S}}}

\newcommand{\ETH}{Exponential Time Hypothesis}

\usepackage{etoolbox}
\usepackage{authblk}

\begin{document}
\title{The Parameterized Complexity of Centrality Improvement in Networks
  \thanks{MS supported by the People Programme (Marie Curie Actions) of the European Union's Seventh Framework Programme (FP7/2007-2013) under REA grant agreement number {631163.11} and the Israel Science Foundation (grant no. 551145/14).}}

\author[1]{Clemens~Hoffmann}

\affil[1]{Institut f\"{u}r Softwaretechnik und Theoretische Informatik,
 TU Berlin, Germany,
 \texttt{h.molter@tu-berlin.de}}

\author[1]{Hendrik~Molter}

\author[2]{Manuel~Sorge}

\affil[2]{Dept.\ Industrial Engineering and Management, Ben-Gurion University of the Negev, Beer Sheva, Israel,
\texttt{sorge@post.bgu.ac.il}}
\maketitle

\begin{abstract}
  \looseness=-1 The centrality of a vertex~$v$ in a network intuitively captures how
  important $v$ is for communication in the network.  The task of
  improving the centrality of a vertex has many applications, as a
  higher centrality often implies a larger impact on the network or
  less transportation or administration cost. In this work we study
  the parameterized complexity of the NP-complete problems \CI\ and
  \BI\ in which we ask to improve a given vertex' closeness or
  betweenness centrality by a given amount through adding a given
  number of edges to the network. Herein, the closeness of a vertex~$v$
  sums the multiplicative inverses of distances of other vertices to $v$
  and the betweenness sums for each pair of vertices the fraction of
  shortest paths going through~$v$. 
  Unfortunately, for the natural parameter ``number of edges to add''
  we obtain hardness results, even in rather restricted cases. On the
  positive side, we also give an island of tractability for the
  parameter measuring the vertex deletion distance to cluster
  graphs.
\end{abstract}

\section{Introduction}

Measuring the centrality of a given vertex in a network has attracted the interest of researchers since the second half of the 20th century~\cite{FREEMAN1978215}, see Newman's book~\cite{newman_networks:_2010} for an overview. There are various interpretations of what makes a vertex more central than another vertex in a network. Two popular measures for the centrality of a vertex~$z$ 
are \emph{closeness centrality}~$c_z$ and \emph{betweenness
centrality}~$b_z$~\cite{FREEMAN1978215}. They are based on the distances of the
given vertex~$z$ to the remaining vertices and on the number of shortest paths going
through~$z$, respectively. For this work, we use the following definitions.
\begin{align*}
c_z = \sum_{\substack{u \in V \\
        d(u,z) < \infty \\ u \neq z}} \frac{1}{d(z,u)}   
  &&
b_z = \sum_{\substack{s,t \in V \\
        s \neq t;s,t \neq z \\
        \sigma_{st} \neq 0}} \frac{\sigma_{stz}}{\sigma_{st}}   
\end{align*}
Herein, $d(s, t)$ is the distance between two vertices~$s$ and $t$, that
is, the number of edges on a shortest $s$-$t$ path, $\sigma_{st}$ is
the number of shortest $s$-$t$ paths, and $\sigma_{stz}$ is the number
of shortest $s$-$t$ paths that contain~$z$. Intuitively, if $z$ has
many close-by vertices, then its closeness centrality is large, and
if $z$ is on shortest paths between many vertices, then
its betweenness centrality is large. The closeness centrality as defined above is also known as the \emph{harmonic centrality}.\footnote{There are several definitions for closeness centrality in the literature. We use the present one because it is natural~\cite{newman_networks:_2010} and it was used in closely related work~\cite{crescenzi2016greedily}.}

\looseness=-1 Analyzing vertex centrality in networks has been studied intensively (e.g. \cite{FREEMAN1978215,okamoto2008ranking,newman2005measure,brandes2008variants,newman_networks:_2010}) and comprises a diverse set of applications in, e.g., biological~\cite{rubinov2010complex}, economic~\cite{Opsahl2010245}, and social networks~\cite{FREEMAN1978215}. 
Some examples: A transport company might be interested in placing its depots centrally such that the transportation costs are rather low. 
The value of an airport might be influenced by its centrality in the flight-connection network between airports. 
The most central nodes in a computer network may be useful for determining the locations of data centers where the routes are short and peering costs are low. In social networks, economically important influencers are presumably more central than other users. 

\looseness=-1 Since it is so desirable to find vertices with large centrality in a graph,
vertices have incentive to improve their own centrality. E.g., a social network
member might want to increase her impact on other users by increasing her own
centrality, or an airport operator wants to increase the appeal of her airport for investors (as measured by the centrality). 
In both cases, natural operations are to introduce new links into the network, i.e., to make new acquaintances or incentivise airlines to offer certain routes. In this work, we hence study the complexity of improving the centrality of a given vertex by introducing new links into the network. Formally, the computational problems that we study are defined as follows.

\decprob{Closeness (Betweenness) Improvement}{An undirected, unweighted graph $G = (V,E)$, a vertex $z \in V$, an integer $k$ and a rational number $r$.}{Is there an edge set $S$, $S \cap E = \emptyset$, of size at most $k$ such that $c_z \geq r$ ($b_z \geq r$) in $G + S :=(V,E \cup S$)?}
We also say that an edge set~$S$ as above is a \emph{solution}.

The above two problems were introduced by
\citet{crescenzi2016greedily} and \citet{d2016maximum}, respectively,
who gave approximation algorithms and showed that their empirical
approximation ratios are close to one on random graphs with up to 100
vertices and up to 1000 vertices, respectively. In a corresponding
presentation \citet{crescenzi2016greedily} noted that finding the
optimal solution for comparison was very time consuming. Here, we
study the parameterized complexity of \CI\ and \BI\ with the ultimate
goal to design efficient exact algorithms. That is, we aim to find
\emph{fixed-parameter (FPT) algorithms} with running time
$f(k) \cdot n^{O(1)}$, where $n$ is the input length and $k$ is some
secondary measure, called \emph{parameter}, or we show W[1] or
W[2]-hardness, meaning that there are presumably no FPT algorithms.

\paragraph{Our Results.} Our results for \CI\ are as follows. 
From two reductions from \DS\ it follows that \CI\ is NP-hard on (disconnected)
planar graphs with maximum degree~3 and W[2]-hard with respect to~$k$, the
number of added edges, on disconnected split graphs, for example
(\cref{corollary:ds-implications}).  Split graphs are a simple model of
core-periphery structure, which occurs in social and biological
networks~\cite{csermely_structure_2013}. In particular, we can derive that a
straightforward $n^{O(k)}$-time algorithm for \CI\ is asymptotically optimal.
Motivated by the fact that social networks often have small diameter in conjunction with small H-index~\cite{eppstein_h-index_2012}, we show that \CI\ remains NP-hard on (connected) graphs of diameter at most~6 and H-index~4 (\cref{theorem:closenesswhard1diameter}). On the positive side, we show that \CI\ allows a fixed-parameter algorithm with respect to the parameter \emph{distance to cluster graph}, that is, the smallest number of vertices to delete in order to obtain a cluster graph. \textsc{Directed Closeness Improvement} is NP-hard and W[2]-hard with respect to~$k$ even if the input graph is acyclic (\cref{theorem:closenesswharddirected}) or has diameter~4 (\cref{theorem:closenessdirecteddiameter}). 

\looseness=-1 For \BI\ the picture is similar. It is W[2]-hard with respect to~$k$ (\cref{theorem:betweennessundirected}) also in the directed case (\cref{theorem:betweennesswharddirected}), NP-hard for graphs of H-index~4 (\cref{corollary:betweennesshindex}), and \BI\ is fixed-parameter tractable with respect to~$k$ and the distance to cluster graph combined. 
\todo[inline]{A point missing in the paper is whether the graphs induced by the NP hardness proofs are practical and in what they represent realistic social network models.

vertex distance to cluster graph does not seem to be very natural and easy to check parameter. 
}

\paragraph{Preliminaries and Notation.}
We use standard notation from graph theory~\cite{reinhard_diestel_graph_2016}. Throughout, we refer to the number of vertices as $n$ and to the number of edges (arcs) as~$m$.  For two vertices $u,v$ we denote by $d(u,v)$ the \emph{distance} between $u$ and $v$, i.e.\ the number of edges on a shortest path from $u$ to~$v$. If $u$ and $v$ are not connected by a path, then $d(u,v) = \infty$. A \emph{split graph} allows for a partition of the vertex set into a clique and an independent set. In a \emph{cluster graph} each connected component is a clique. The \emph{diameter} of a graph is $\infty$ if it is disconnected and the maximum distance of any two vertices otherwise. The \emph{H-index} of a graph is the largest integer~$h$ such that there $h$ vertices of degree at least~$h$. 

\looseness=-1 We also use standard notation from parameterized complexity~\cite{cygan2015parameterized}. Importantly, a \emph{parameterized reduction} from a parameterized problem $L \subseteq \Sigma^* \times \mathbb{N}$ with parameter $k$ to a parameterized problem~$L' \subseteq \Sigma^* \times \mathbb{N}$ with parameter~$k'$ is a
$g(k) \cdot |I|^{\mathcal{O}(1)}$-time computable function~$f: \Sigma^* \times \mathbb{N} \rightarrow \Sigma^* \times \mathbb{N}: (I,k) \rightarrow (I',k')$ such that $k' \leq h(k)$ for some computable function $h$ and $(I,k) \in L \Leftrightarrow (I',k') \in L'$.

The \emph{Exponential Time Hypothesis} roughly states that
satisfiability of a Boolean formula in conjunctive normal form with
clauses of size~3 cannot be decided in $2^{o(n)}$~time, see
\citet{impagliazzo_which_1998,impagliazzo_complexity_1999} for details.

\section{Closeness Centrality}\label{chapter:closeness}

In this section, we present algorithmic and hardness results for
\textsc{Closeness Improvement}. First, we make an important observation that
will help us in our proofs. Intuitively, we show that to improve the closeness
of a vertex by adding edges, it always makes sense to add only edges adjacent to that
vertex.

\begin{lemma}\label{lemma:closenessendpointz}
	Let~$I = (G = (V,E), z,k, r)$ be a \textsc{Closeness Improvement} instance. If~$I$ is a \YI{}, then~$c_z$ can be increased to~$r$ by adding at most~$k$ edges, all of which contain~$z$.  
\end{lemma}
\begin{proof}
  Let~$u, v \in V \setminus \{z\}$ be two vertices of the input graph~$G$
  such that~$e := \{u, v\} \notin E$. Assume for the sake of
  contradiction that $S$ is a solution with the largest number of
  edges incident to~$z$ and that $\{u, v\} \in S$. Consider the
  shortest paths from $u$ to $z$ and from $v$ to $z$ in $G + S$. If
  these paths have the same length, then neither contains~$\{u,
  v\}$. Hence, in this case we have that $c_z$ in $G + S$ equals $c_z$
  in $G + (S \setminus\{u, v\})$. Thus, adding an arbitrary edge to
  $(S \setminus \{u, v\}$ yields a solution with a greater number of
  edges incident to $z$ than~$S$, a contradiction. Hence, one of the
  two shortest paths is shorter than the other; say $u$ is closer to
  $z$ than~$v$. (Observe that, hence, $\{v, z\} \notin E$.) In this
  case, $c_z$ in $G + ((S \setminus \{u, v\}) \cup \{v, z\})$ is at
  least as large as $c_z$ in $G + S$: Consider any shortest path~$P$
  in $G + S$ from some vertex $w$ to $z$ that contained $\{u,
  v\}$. Since $u$ is closer to $z$ than $v$, path~$P$ contains first
  $v$ and then~$u$. Hence, replacing the remaining path after $v$ with
  the direct edge to~$z$ shortens~$P$. Hence,
  $((S \setminus \{u, v\}) \cup \{v, z\})$ contains more edges
  incident to~$z$ and yields $c_z$ which is at least large as for $S$,
  a contradiction. Hence, the solution with the largest number of
  edges incident to~$z$ does not contain any edges not incident
  to~$z$, showing the lemma.
 \end{proof}

From this observation  follows, that if we were to try all possibilities to
 solve \textsc{Closeness Improvement}, it suffices to consider only sets of
 edges to add, where every edge is adjacent to the node whose closeness we want
 to improve. Hence, 
we get an XP algorithm with respect to $k$.

\begin{corollary} \label{corollary:CIXP}
    \textsc{Closeness Improvement} is solvable in $O(n^k \cdot (n + m))$ time where
    $k$ is the number of edge additions, and thus is in XP with respect to the
    parameter number $k$ of edge additions.
\end{corollary}

\paragraph{Hardness Results.} \label{section:closeness_negative}

Next, we present several hardness results for \textsc{Closeness
  Improvement}: They are based on two reductions from \DS; a simple
one and a more intricate one. From results on \DS\ we can then infer
corresponding results for \CI. In particular, we show that the
$n^{O(k)}$-time algorithm from \cref{corollary:CIXP} is essentially
optimal unless the \ETH\ is false.

\begin{figure}[t!]
  \centering
  \hfill
	\begin{tikzpicture}[scale=1,style=transform shape]
	\tikzstyle{knoten}=[circle,draw,fill,minimum size=5pt,inner sep=2pt,style=transform shape]
	\node[knoten,label={[label distance=0cm]90:$u_1$}] (K-1) at (0,0) {};
	\node[knoten,color=red,label={[label distance=0cm]90:$u_2$}] (K-2) at (1,0) {};
	\node[knoten,color=red,label={[label distance=0cm]270:$u_3$}] (K-3) at (0,-1) {};
	\node[knoten,label={[label distance=0cm]270:$u_4$}] (K-4) at (1,-1) {};
	\node[knoten,label={[label distance=0cm]90:$u_5$}] (K-5) at (2,0) {};
	\node[knoten,label={[label distance=0cm]270:$u_6$}] (K-6) at (2,-1) {};
	
	% \node[knoten,color=white,label={[label distance=0cm]270:$$}] (K-55) at (2,-1.8) {};
	% Connect nodes
	\foreach \i / \j in {1/3,2/4,2/5,2/6, 3/4}{
		\path (K-\i) edge[-] (K-\j);
	}
	\end{tikzpicture}
  \hfill
\begin{tikzpicture}[scale=1,style=transform shape]
	\tikzstyle{knoten}=[circle,draw,fill,minimum size=5pt,inner sep=2pt,style=transform shape]
	\node[knoten,label={[label distance=0cm]90:$u_1$}] (K-1) at (0,0) {};
	\node[knoten,label={[label distance=0cm]90:$u_2$}] (K-2) at (1,0) {};
	\node[knoten,label={[label distance=0cm]270:$u_3$}] (K-3) at (0,-1) {};
	\node[knoten,label={[label distance=0cm]270:$u_4$}] (K-4) at (1,-1) {};
	\node[knoten,label={[label distance=0cm]90:$u_5$}] (K-5) at (2,0) {};
	\node[knoten,label={[label distance=0cm]270:$u_6$}] (K-6) at (2,-1) {};
	\node[knoten,label={[label distance=0cm]270:$z$}] (Z) at (4,-.5) {};
	% Connect nodes
	\foreach \i / \j in {1/3,2/4,2/5,2/6, 3/4}{
		\path (K-\i) edge[-] (K-\j);
	}

	\path (Z) edge[-,color=red, dashed] (K-2);
	\path (Z) edge[-,color=red, dashed] (K-3);
	\end{tikzpicture}
  \hfill\mbox{}
  \caption{Reduction from \DS{} to \textsc{Closeness Improvement}. Left: A \DS{} instance $I = (G,k=2)$. The red colored vertices form a solution for $I$. Right: The constructed \textsc{Closeness Improvement} instance ($I' = G',z,k'=2,r = k+\frac{n-k}{2}$). The red dashed edges form a solution for $I'$.}
  \label{figure:red_dom_closeness1}
\end{figure}
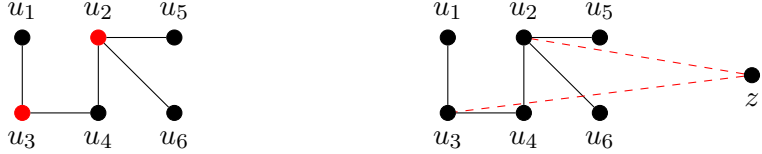

\begin{theorem}\label{theorem:closenesswhard1}
  \textsc{Closeness Improvement} is NP-hard and W[2]-hard with respect to the number~$k$ of edge additions (on disconnected graphs). Moreover, unless the \ETH\ fails, \CI\ does not allow an algorithm with running time~$f(k)\cdot n^{o(k)}$.
\end{theorem}
  \begin{proof}
    We give a reduction from \DS{} which is W[2]-hard and does not admit an algorithm with running time~$f(k)\cdot n^{o(k)}$ unless the \ETH\ is false~\cite{cygan2015parameterized}. Let~$I = (G = (V,E), k \in \mathbb{N})$ where~$V = \{u_1,\ldots,u_n\}$. We construct a \textsc{Closeness Improvement} instance~$I' = (G' = (V',E'), z, k, k + \frac{1}{2}(n-k))$ as follows: Given the input graph~$G$, we simply add an isolated vertex $z$ to the graph, that is $G' = (V \cup \{z\}, E)$.
    
   We now show that the reduction is correct, i.e.~$I$ is a \YI{} if and only if~$I'$ is a \YI{}:
    
   $\Rightarrow$: Let there be a dominating set~$V_{DS} \subseteq V$ of size~$k$ in~$G$. After adding~$k$ edges between~$z$ and each vertex in~$V_{DS}$ in~$G'$, these~$k$ vertices have distance 1 to~$z$ and the~$n-k$ neighbors of the vertices~$V_{DS}$ have distance 2 to~$z$. Hence,~$c_z = k + \frac{n-k}{2}$. That is, $I'$~is a \YI{}.
    
   $\Leftarrow$: The reverse direction is by contraposition. 
   Assume that there is no dominating set of size~$k$ in~$G$. \cref{lemma:closenessendpointz} shows that we can maximally increase $c_z$ by adding edges where one endpoint is $z$. However, after adding~$k$ edges between~$z$ and $k$ other vertices in~$G'$, there are $\ell \geq 1$ vertices in~$G'$ whose distances to~$z$ is $d \geq 3$. Hence,~$c_z \leq k + \frac{n-k-\ell}{2} + \frac{\ell}{d} < k +  \frac{n-k}{2}$. That is,~$I'$ is a \NI{}.
 \end{proof}
  \begin{corollary}\label{corollary:ds-implications}
    \CI\ is
  \begin{compactenum}
  \item NP-hard even on disconnected planar graphs with maximum degree~3,
  \item NP-hard and W[2]-hard on disconnected split graphs, and
  \item NP-hard and W[2]-hard on disconnected graphs in which each connected component has diameter two.
  \end{compactenum}
\end{corollary}
\begin{proof}
  Instead of the plain \DS\ problem, we reduce from special cases which have been shown to be hard in the literature.   
  For the first result, we reduce from \DS\ on planar graphs with maximum degree~3~\cite{Garey:1990:CIG:574848}. Clearly, the reduction in \cref{theorem:closenesswhard1} does not destroy planarity or increase the maximum degree.

  For the second result, we use the standard reduction from \SC\ to \DS\ in split graphs, which identifies the size of the sought set cover and the size of the sought dominating set~\cite{bertossi_dominating_1984}. Since \SC\ is W[2]-hard with respect to the size of the sought set cover~\cite{downey2012parameterized}, so is \DS\ on split graphs with respect to the size of the sought dominating set. Hence, the same is true for \CI\ on split graphs by the reduction in \cref{theorem:closenesswhard1}.

  For the third result, we use the fact that \DS\ is W[2]-hard on graphs of diameter~2~\cite{DBLP:conf/iwpec/LokshtanovMPRS13}.
\end{proof}

\noindent Applications in social networks, which have both small diameter and small H\nobreakdash-index\footnote{Recall that the H-index of a graph is the largest integer~$h$ such that there are $h$ vertices of degree at least~$h$.}~\cite{eppstein_h-index_2012}, motivate the following more special hardness result where both these values are small constants. It also shows that \CI\ remains hard on connected graphs, which was left open by \cref{theorem:closenesswhard1}. 

\begin{figure}[t!]
  \centering
  \hfill
  \raisebox{-0.5\height}{
	\begin{tikzpicture}[scale=1,style=transform shape]
	\tikzstyle{knoten}=[circle,draw,fill,minimum size=5pt,inner sep=2pt,style=transform shape]
	\node[knoten,label={[label distance=0cm]90:$u_1$}] (K-1) at (0,0) {};
	\node[knoten,color=red,label={[label distance=0cm]90:$u_2$}] (K-2) at (1,0) {};
	\node[knoten,color=red,label={[label distance=0cm]270:$u_3$}] (K-3) at (0,-1) {};
	\node[knoten,label={[label distance=0cm]270:$u_4$}] (K-4) at (1,-1) {};
	\node[knoten,label={[label distance=0cm]90:$u_5$}] (K-5) at (2,0) {};
	\node[knoten,label={[label distance=0cm]270:$u_6$}] (K-6) at (2,-1) {};
	
	\foreach \i / \j in {1/3,1/2,2/4,2/5,2/6, 3/4}{
		\path (K-\i) edge[-] (K-\j);
	}
      \end{tikzpicture}
      }
      \hfill
      \raisebox{-0.5\height}{
\begin{tikzpicture}[scale=1,style=transform shape]
	\tikzstyle{knoten}=[circle,draw,fill,minimum size=5pt,inner sep=2pt,style=transform shape]
	\node[knoten,label={[label distance=0cm]90:$u_1$}] (K-1) at (0,0) {};
	\node[knoten,label={[label distance=0cm]90:$u_2$}] (K-2) at (1,0) {};
	\node[knoten,label={[label distance=0cm]270:$u_3$}] (K-3) at (0,-1) {};
	\node[knoten,label={[label distance=0cm]270:$u_4$}] (K-4) at (1,-1) {};
	\node[knoten,label={[label distance=0cm]90:$u_5$}] (K-5) at (2,0) {};
	\node[knoten,label={[label distance=0cm]270:$u_6$}] (K-6) at (2,-1) {};
	\node[knoten,label={[label distance=0cm]270:$z$}] (Z) at (7,-.5) {};
	% Connect nodes
	\foreach \i / \j in {1/2,1/3,2/4,2/5,2/6, 3/4}{
		\path (K-\i) edge[-] (K-\j);
	}
	
	\foreach \i in {1,2,3,4,5,6}{
		\node[knoten,label={[label distance=0cm]270:$x_\i$}] (X-\i) at (4, -\i+3) {};
	}
	
	\foreach \i in {1,2,3,4,5,6}{
		\node[knoten,label={[label distance=0cm]270:$y_\i$}] (Y-\i) at (5, -\i+3) {};
	}
	
	\foreach \i in {1,2,3,4,5,6}{
		\path (K-\i) edge[-,bend left] (X-\i);
		\path (X-\i) edge[-] (Y-\i);
		\path (Y-\i) edge[-] (Z);
	}
	
	\path (Z) edge[-,color=red, dashed] (K-2);
	\path (Z) edge[-,color=red, dashed] (K-3);
	\end{tikzpicture}
        }
  \hfill\mbox{}
    \caption{Parameterized reduction from \DS{} to \textsc{Closeness Improvement} on graphs with diameter 4. Left: A \DS{} instance $I = (G,k=2)$. The red colored vertices $u_2$ and $u_3$ form a solution for $I$. Right: The constructed \textsc{Closeness Improvement} instance ($I' = G',z,k'=2,2n+\frac{k}{2}$). The red, dashed edges form a solution for $I'$.}
    \label{figure:red_dom_closeness}
\end{figure}
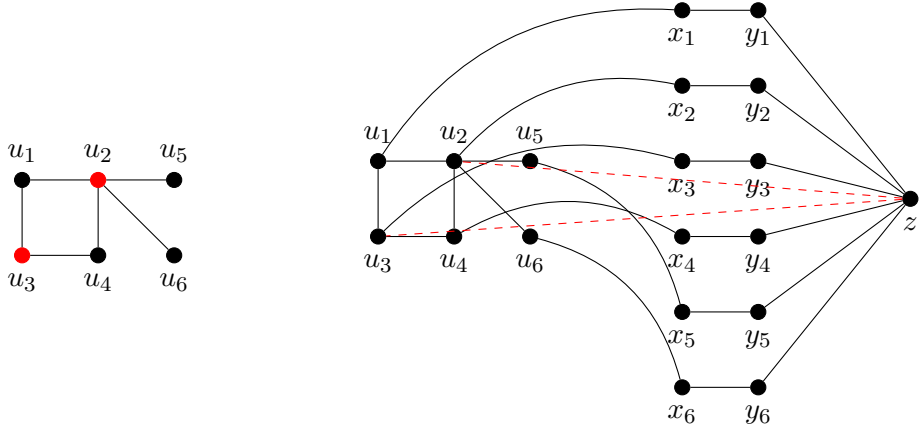

\begin{theorem}
  \label{theorem:closenesswhard1diameter}
  \textsc{Closeness Improvement} is NP-hard and W[2]-hard with respect to the parameter number~$k$ of edge additions even on connected graphs with diameter~4. Moreover, \CI\ is NP-hard even on graphs which simultaneously have diameter~6 and H-index~4.
\end{theorem}
\begin{proof}
    The proof is by a parameterized reduction from two variants of \DS{}, detailed below. Let~$I = (G = (V,E), k \in \mathbb{N})$ be a
    \DS{} instance where~$V = \{u_1,\ldots,u_n\}$. We construct a \textsc{Closeness Improvement} instance~$I'
    = (G' = (V',E'), z, k, 2n + \frac{k}{2})$ as follows (see
    \cref{figure:red_dom_closeness}): Given the input graph~$G$, we add~$2n$
    vertices~$x_1,\ldots,x_n,y_1,\ldots,y_n$ such that each vertex~$x_i$ is adjacent
    to~$u_i$ and~$y_i$. Furthermore, we add~$z$ and add edges between~$z$ and
    each~$y_1,\ldots,y_n$. Formally, $V' =V \cup \{x_i,y_i \mid 1 \leq i \leq n\} \cup \{z\}$ and~$E' = E \cup \{\{u_i,x_i\}, \{x_i,y_i\}, \{z,y_i\} \mid 1 \leq i \leq n\}$. 
    
    We partition~$V'$ into the subsets~$Y' := \{y_1,\ldots,y_n\}, X'
    := \{x_1,\ldots,x_n\}$ and~$U' :=
    \{u_1,\ldots,u_n\}$. Note that the vertices
    in~$Y'$ have distance 1 to~$z$, the vertices
    in~$X'$ have distance 2 to~$z$ and the vertices
    in~$U'$ all have distance 3
    to~$z$.  This completes the construction which can clearly be
    carried out in polynomial time.

    Suppose the reduction is correct. To get NP-hardness and
    W[2]-hardness with respect to~$k$ on diameter~4 graphs, we reduce
    from \DS\ on graphs of diameter two, which is
    NP-hard~\cite{ambalath_kernelization_2010} and W[2]-hard with
    respect to~$k$~\cite{DBLP:conf/iwpec/LokshtanovMPRS13}. It is not
    hard to see that the resulting graph~$G'$ indeed has
    diameter~$4$. To get NP-hardness on graphs with simultaneously
    diameter~$6$ and H-index~4, reduce instead from \DS\ on graphs~$G$
    with maximum degree~4. By the connections via~$x_i$, $y_i$, and $z$,
    any two vertices of~$G$ are connected in~$G'$ by a path of length at
    most~6. Graph~$G'$ has H-index $4$, because it has maximum degree~4 apart from~$z$.
    
First, we show that adding edges
    between~$z$ and vertices in~$U'$ is optimal:
    
    Assume an edge~$\{z,x_i\}, x_i \in X'$, is added. Then the distance
    between~$z$ and~$x_i$ is 1 and the distance between~$z$ and~$u_i$ is 2. If
    we instead add the edge~$\{z,u_i\}$, then the distance between~$z$
    and~$u_i$ is 1 and the distance between~$z$ and~$x_i$ remains 2.
    Furthermore, the edge~$\{z,u_i\}$ may introduce shorter distances to the
    neighbors of~$u_i$, which the edge~$\{z,x_i\}$ does not. Hence, if a
    solution for~$I'$  contains~$\{z,x_i\}$, we can replace that edge
    by~$\{z,u_i\}$.
    
It remains to show that the reduction is correct, i.e.~$I$ is a \YI{} if and only if~$I'$ is a \YI{}:
    
    $\Rightarrow$: Let~$I$ be a \YI{}. Then there is a dominating set~$U_{DS}'
    \subseteq U'$ of size~$k$ for~$G' - (X'\cup Y')$. After adding~$k$ edges
    between~$z$ and each vertex in~$U_{DS}'$, these~$k$ vertices have distance 1
    to~$z$ and the~$n-k$ neighbors in~$U' \setminus U_{DS}'$ have distance 2
    to~$z$. Furthermore, each vertex in~$Y'$ has distance 2 to~$z$ and each vertex
    in~$X$ has distance 1 to~$z$.
    
    Hence,~$c_z = k + \frac{n-k}{2} + \frac{n}{2} + n = 2n + \frac{k}{2}$. That
    is,~$I'$ is a \YI{}.
    
    $\Leftarrow$: We prove the other way by contraposition. That is, we show
    that~$I'$ is a \NI{} if~$I$ is a \NI{}. Let~$I$ be a \NI{}. If there is no
    dominating set of size~$k$ in~$G' - (X' \cup Y')$, then after adding~$k$
    edges between~$z$ and vertices in~$U'$, there are~$l \geq 1$ vertices in~$U'$ whose distances to~$z$ is still 3.
    
    Hence,~$c_z = k + \frac{n-k-l}{2} + \frac{l}{3} + \frac{n}{2} + n  = 2n +
    \frac{k}{2} - \frac{l}{6}$, for~$l \geq 1$. That is,~$I'$ is a \NI{}. 
  \end{proof}
We note that it is not hard to show that \CI\ is polynomial-time solvable on graphs of diameter~2. The case of diameter~3 remains open.

\paragraph{Algorithmic Result.}
Now we present an algorithm for \textsc{Closeness Improvement}, which shows that
the problem is fixed-parameter tractable when parameterized by the distance of
the input graph to a cluster graph.

\begin{theorem} \label{theorem:distancetoclustercloseness}
    \textsc{Closeness Improvement} can be solved in~$2^{2^{2^{O(\ell)}}} \cdot
    n^{O(1)}$~time, where~$\ell$ is the vertex deletion distance of~$G$ to a
    cluster graph, and thus is in FPT with respect to the parameter~$\ell$.
\end{theorem}
\newcommand{\VDS}{\textsf{VDS}}
\newcommand{\sig}{\textsf{sig}}
\begin{proof}
  \looseness=-1 Let~$(G,z,k,r)$ be a \textsc{Closeness Improvement} instance, where $V_{\VDS} \subset V$~is a vertex set of size~$\ell$ such that~$G_C = (V_C,E_C) := G - V_{\VDS}$ is a cluster graph with the set of connected components~$\C = \{C_1, \ldots, C_s\}$ which we also call \emph{clusters}. Since a cluster vertex deletion set~$V_{\VDS}$ of size~$\ell$ can be found in $O(1.92^\ell\cdot(n + m))$~time if it exists~\cite{boral_fast_2016,huffner2010fixed}, we may assume that $V_{\VDS}$ is given. By \cref{lemma:closenessendpointz} we may assume that the edges in an optimal solution~$E^*$ to $(G, z, k, r)$ all have endpoint~$z$. Hence, in the following we denote by a solution~$V'$ the endpoints different from~$z$ of the corresponding edge set. Any solution can thus be divided into vertices in $V_{\VDS}$ and those in $V \setminus V_\VDS$. Let $V^*_{\VDS}$ be the intersection of an optimal solution~$V^*$ with $V_\VDS$. The first step in our algorithm is to find~$V^*_\VDS$, by trying all $2^\ell$ possibilities. It remains to determine $V^* \setminus V_\VDS$. Intuitively, if there are vertices which have the same neighborhood in $V_\VDS$ and are in clusters that also have the same neighborhood in $V_\VDS$, then each such vertex after the first one does not help to shorten distances to $z$ for any vertex except itself. Hence, if we know that the optimal solution contains vertices in clusters both with some specified neighborhood in~$V_\VDS$, then we can assume that these vertices are distributed among the largest clusters with that neighborhood. In the algorithm we thus first determine for which neighborhoods in $V_\VDS$ there are clusters and vertices in these clusters in the optimal solution. Then we distribute the vertices in the solution optimally among the chosen neighborhoods. The proof that this yields an optimal solution is unfortunately technical and we need the following notation.

  We say that the \emph{signature $\sig(C_i)$ of a cluster}~$C_i$, $i = 1, \ldots, s$, is the set of neighbors in~$V_\VDS \cup \{z\}$ of vertices in $C_i$, that is, the signature is $\{v \in V_\VDS \cup \{z\} \mid \exists u \in C_i \colon \{u, v\} \in E\}$. Similarly, the \emph{signature $\sig(v)$ of a vertex} $w \in V \setminus V_\VDS$ is $N(v) \cap (V_\VDS \cup \{z\})$. For some subset $V_i \subseteq C_i$ of some cluster~$C_i \in \C$ denote by the \emph{signature $\sig(V_i)$ of $V_i$} the tuple $(\sig(C_i), \{\sig(v) \mid v \in V_i\})$. Say also that $C_i$ is \emph{$V_i$'s cluster}.
  Now the \emph{signature~$\sig(\hat{V})$ of a solution}~$\hat{V}$ is the set $\{\sig(V_i) \mid C_i \in \C \wedge C_i \cap \hat{V} = V_i \neq \emptyset\}$.
  That is, the signature of~$\hat{V}$ encodes the signatures of the clusters touched by~$\hat{V}$ along with, for each touched cluster, the signatures of all vertices touched by~$\hat{V}$ in that cluster. Say that a vertex subset~$V_j$ of some cluster~$C_j$ is \emph{eligible} for some signature~$\sig(V_i)$ of a vertex subset~$V_i$ of a cluster~$C_i$ if $\sig(V_j) = \sig(V_i)$. Accordingly, for some solution $\hat{V}$ with signature $\sig(\hat{V})$, say that a vertex subset~$V_i \subseteq C_i$ of some cluster $C_i$ is \emph{eligible} for $\sig(\hat{V})$ if $\sig(V_i) \in \sig(\hat{V})$. 
  Finally, the \emph{reduct} of a solution $\hat{V}$ is a subset~$V' \subseteq \hat{V}$ such that, for each cluster~$C_i \in \C$ with $\hat{V} \cap C_i \neq \emptyset$ and each vertex signature $S \in \{\sig(v) \mid v \in \hat{V} \cap C_i\}$, there is exactly one vertex $u \in V' \cap C_i$ with signature~$S$. Observe that, if $V'$ is the reduct of $\hat{V}$ and $V'' \supseteq V'$ is any superset of~$V'$ with $|V''| = |\hat{V}|$, then $\sig(\hat{V}) = \sig(V') \subseteq \sig(V'')$ and the closeness centrality of~$z$ achieved by $V''$ is at least the one achieved by $\hat{V}$.

  Let $S$ be the signature of some vertex subset of some cluster. Call a vertex subset~$V_i \subseteq C_i$ of some cluster~$C_i \in \C$ \emph{most potent} for~$S$ if it is eligible for~$S$ and among all vertex subsets of some cluster in~$\C$ that are eligible for~$S$ we have that $V_i$'s cluster is the largest. If the signature~$S$ is clear from the context, we say that~$V_i$ is most potent.

  \looseness=-1 Let $V^*$ be the reduct of an optimal solution. We claim that there is an optimal solution with reduct~$V^*_2$ with signature $\sig(V^*_2) = \sig(V^*)$ such that, for each $S \in \sig(V^*_2)$, there is a vertex subset~$V_j$ of some cluster contained in~$V^*$ that is most potent among vertex subsets eligible for~$S$. Assume the claim does not hold. Then there exists the reduct~$V^*_3$ of some optimal solution such that $V^*_3$ contains the largest number of most potent vertex sets and at least one signature~$S \in \sig(V^*_3)$ such that no vertex subset of some cluster which is most potent for~$S$ is contained in~$V^*_3$. Observe however, that some vertex subset $V_{i} \subseteq C_i$ with $\sig(V_i) = S$ is contained in~$V^*_3$. Let $V_j$ be most potent among vertex sets with signature~$\sig(V_{i})$ and let $V^*_4 = (V^*_3 \setminus V_i) \cup V_j$. Note that, $\sig(V^*_4) = \sig(V^*_3)$ and, because of that, each vertex in $V \setminus (C_i \cup C_j)$ has the same distance to~$z$ according to $V^*_3$ and to $V^*_4$. However, since~$|C_i| < |C_j|$, $\sig(V_i) = \sig(V_j)$, and since, for each vertex signature in $\sig(V_i)$ there is at most one vertex in each of~$V^*_4$ and $V^*_3$ with that signature, more vertices have distance~2 to~$z$ according to~$V^*_4$ than to~$V^*_3$. This is a contradiction to $V^*_3$ being the reduct of an optimal solution. Hence, the claim holds. Thus, once we know the signature of an optimal solution, we know it is optimal to take the most potent (according to that signature) vertex sets into our solution.

  Let $V^*$ again be the reduct of an optimal solution. The \emph{remainder} of $V^*$ is the subset of $V^*$ resulting from removing for each $S \in \sig(V^*)$ a most potent vertex set~$V_j$ with signature~$\sig(V_j) = S$ from~$V^*$ (note that the $V_j$'s are present without loss of generality by the previous claim).

  We claim that there is some optimal solution with reduct~$V_2^*$ with signature~$\sig(V^*_2) = \sig(V^*)$ such that the remainder of $V_2^*$ contains among all vertex subsets of some cluster with a signature in $\sig(V^*_2)$ those vertex subsets in the largest clusters. Assume otherwise. Then there exists the reduct~$V^*_3$ of some optimal solution such that the remainder~$V^{*, R}_3$ of $V^*_3$ contains some $V_{i} \subseteq C_{i}$ such that $V^{*, R}_3 \cap C_i = V_i$, and there is a cluster~$C_j$ and a vertex subset~$V_j \subseteq C_j$ such that~$V^{*, R}_3 \cap V_j = \emptyset$, $\sig(V_j), \sig(V_i) \in \sig(V^*_3)$, and $|C_j| > |C_i|$. Let $V^*_4 = (V^*_3 \setminus V_i) \cup V_j$. Note that $\sig(V^*_3) = \sig(V^*_4)$ and, because of that, each vertex in $V \setminus (C_i \cup C_j)$ has the same distance to~$z$ according to $V^*_3$ and to~$V^*_4$. By the same argument as in the previous claim, we obtain a contradiction to $V^*_3$ being the reduct of an optimal solution. Hence, also this claim holds. Thus, once we know the signature~$S$ of an optimal solution, we know it is optimal to take into the remainder of the optimal solution those vertex subsets with a signature in~$S$ that are contained in the largest clusters.

  The algorithm to compute an optimal solution~$V^*$ is now as follows. Try all possibilities for the intersection~$V^* \cap V_\VDS$. Next, try all possibilities for the signature~$S$ of~$V^*$. Put into~$V^*$, for each $S' \in S$, a vertex subset of some cluster which is most potent for~$S'$. Then, find the smallest vertex subsets of the clusters which have some signature in~$S$ and add them to $V^*$ in decreasing order of the size of their cluster as long as $|V^*| \leq k$. Finally, add to $V^*$ arbitrary vertices until~$|V^*| = k$. This algorithm finds an optimal solution because at least one of the possibilities checked above corresponds to an optimal solution and by the claims above.

  It remains to show the running time: There are at most $2^\ell$
  possibilities for~$V^* \cap V_\VDS$. For each signature of a cluster~$C_i$, of which there are at most~$2^\ell$, there are at most $2^{2^\ell}$ possibilities for the set of vertex signatures of a subset of~$C_i$. Hence, the signature of $V^*$ is the subset of a set of size $2^\ell \cdot 2^{2^\ell}$, meaning that there are at most $2^{2^{2^{O(\ell)}}}$ possibilities for the signature of~$V^*$. Hence, the algorithm checks at most $2^{2^{2^{O(\ell)}}}$ possibilities. To see that the cluster vertex subsets added to~$V^*$ for each possibility can be computed in polynomial time, observe that the it suffices to iterate over each cluster, find its signature and the signature of its vertices and accumulate the largest ones into a dictionary data structure indexed by the size of the clusters. 
\end{proof}

\paragraph{Directed Closeness Improvement.}
\label{section:closenessdirected}

\begin{figure}[t!]
  \centering
  \hfill
	\begin{tikzpicture}[scale=1,style=transform shape]
	\tikzstyle{knoten}=[circle,draw,fill,minimum size=5pt,inner sep=2pt,style=transform shape]
    %\tikzstyle{family}=[ellipse [x radius=1, y radius=2],draw,minimum size=10pt,inner sep=10pt,style=transform shape]
	\node[knoten,label={[label distance=-0.05cm]90:$s_1$}] (K-1) at (1,-1) {};
	\node[knoten,label={[label distance=-0.05cm]90:$s_2$}] (K-2) at (0,-1) {};
	\node[knoten,label={[label distance=-0.05cm]270:$s_3$}] (K-3) at (0,-2) {};
	\node[knoten,label={[label distance=-0.05cm]270:$s_4$}] (K-4) at (0,-3) {};
	\node[knoten,label={[label distance=-0.05cm]270:$s_5$}] (K-5) at (1,-2) {};
    
	%F1
	\draw[color=red,rounded corners=15pt]
    (-0.5,-1.5) rectangle ++(2,1);;
    
    %F2
    \draw[color=blue, fill=blue, opacity=0.1, rounded corners=15pt]
    (-0.5,-3.5) rectangle ++(1,3);;
    
    %F4
    \draw[color=green, fill=green, opacity=0.1, rounded corners=15pt]
    (0.5,-2.5) rectangle ++(1,2);;
    
    %F3
    \draw[rounded corners=15pt]
    (-0.5,-2.5) rectangle ++(2,1);;
    
    \node[label={[color=red,label distance=0.3cm]180:$F_1$}, outer sep=2pt] at (0,-1) {};
    
    \node[label={[label distance=0.3cm]180:$F_3$}, outer sep=2pt] at (0,-2) {};
    
    \node[label={[color=blue,label distance=-0.3cm]270:$F_2$}, outer sep=2pt] at (0,0) {};
    
    \node[label={[color=green,label distance=-0.3cm]270:$F_4$}, outer sep=2pt] at (1,0) {};
    
	% Connect nodes
	%\foreach \i / \j in {1/2,1/3,2/4,2/5,2/6, 3/4}{
	%	\path (K-\i) edge[-] (K-\j);
	%}
	\end{tikzpicture}
  \hfill
	\begin{tikzpicture}[scale=1,style=transform shape]
	\tikzstyle{knoten}=[circle,draw,fill,minimum size=5pt,inner sep=2pt,style=transform shape]
	\node[knoten,label={[label distance=-0.05cm]180:$u_1$}] (U-1) at (0,0) {};
    \node[knoten,label={[label distance=-0.05cm]180:$u_2$}] (U-2) at (0,-1) {};
    \node[knoten,label={[label distance=-0.05cm]180:$u_3$}] (U-3) at (0,-2) {};
    \node[knoten,label={[label distance=-0.05cm]180:$u_4$}] (U-4) at (0,-3) {};
    \node[knoten,label={[label distance=-0.05cm]180:$u_5$}] (U-5) at (0,-4) {};
    
    \node[knoten,label={[label distance=-0.05cm]90:$v_1$}] (V-1) at (1,-0.5) {};
    \node[knoten,label={[label distance=-0.05cm]90:$v_2$}] (V-2) at (1,-1.5) {};
    \node[knoten,label={[label distance=-0.05cm]270:$v_3$}] (V-3) at (1,-2.5) {};
    \node[knoten,label={[label distance=-0.05cm]270:$v_4$}] (V-4) at (1,-3.5) {};
    
    \node[knoten,label={[label distance=-0.05cm]270:$z$}] (Z) at (3.5,-2) {};
	
	%\node[knoten,color=white,label={[label distance=0cm]270:$$}] (K-55) at (2,-1.8) {};
	% Connect nodes

    \path (Z) edge[-latex, dashed, color=red] (V-2);
    \path (Z) edge[-latex, dashed, color=red] (V-4);
    
	\foreach \i / \j in {1/1,2/1,3/3,5/3,2/2,3/2,4/2,1/4,5/4}{
		\path (U-\i) edge[latex-] (V-\j);
	}
	\end{tikzpicture}
  \hfill\mbox{}
    \caption{Parameterized reduction from \textsc{Set Cover} to \textsc{Directed Closeness Improvement}. Left: A \textsc{Set Cover} instance~$I = (U,\mathcal{F},k=2$) with solution~$\{F_2,F_4\}$. Right: The constructed \textsc{Directed Closeness Improvement} instance~$I' = (G,z,k=2,r=4\frac{5}{6})$. The red dashed edges imply a solution for~$I'$.}
    \label{figure:directedcloseness}
\end{figure}
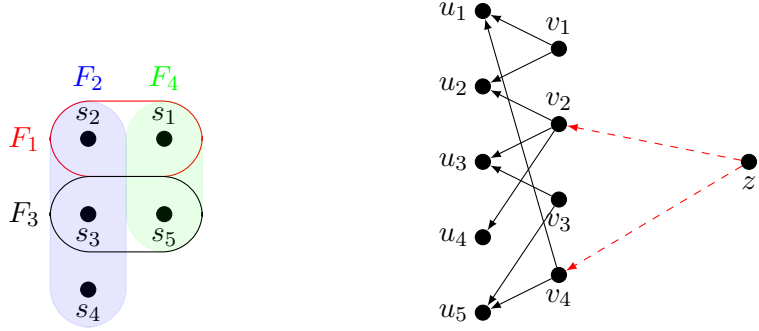

\looseness=-1 We now investigate the problem \textsc{Directed Closeness
Improvement} of improving the closeness centrality of a vertex~$z$ on directed,
unweighted graphs. Herein, the closeness centrality is measured by sum of the
multiplicative inverse distances \emph{from}~$z$ \emph{to} the other
vertices\footnote{It is easy to check that all our results also hold
if the closeness centrality is measured by sum of the
multiplicative inverse distances \emph{from the other
vertices to}~$z$.}. We show that the problem remains W[2]-hard with respect to the number~$k$ of added arcs, even on directed acyclic graphs and even if the diameter of the graph is 4. 
Analogously to the undirected variant, we show that we can maximize the closeness centrality of a vertex $z$ in a directed graph by adding arcs adjacent to~$z$: 
\begin{lemma} \label{lemma:directedclosenessendpoints}
    Let~$I = (G = (V,E), z,k, r)$ be a \textsc{Directed Closeness Improvement} instance. If~$I$ is a \YI{}, then there is a solution $S$ for $I$ where for each arc $a \in S$, the source vertex is~$z$.
\end{lemma}
\begin{proof}
    The proof is analogous to the one of \cref{lemma:closenessendpointz}: If an optimal solution~$S$ contains an arc $a := (u,v), u,v \neq z$, then any shortest path from $z$ to some vertex $w$ containing the arc $a$ becomes even shorter if $(u,v)$ is replaced by $(z,v)$. If $(z,v)$ already exists, then no shortest path from $z$ contains $a$; hence, it can be replaced by an arbitrary arc with source $z$. Furthermore, an arc $a'$ where $z$ is the endpoint does not improve the closeness centrality of $z$ at all, as any path from $z$ containing $a'$ contains a loop and thus is no shortest path. 
\end{proof}

\cref{lemma:directedclosenessendpoints} directly implies that \textsc{Directed Closeness Improvement} is in XP with respect to the number of arc additions:

\begin{corollary} \label{lemma:directedclosenessxp}
    \textsc{Directed Closeness Improvement} can be solved in $O(n^k \cdot (n+m))$ time, where $k$ is the number of arc additions and thus is in XP with respect to the parameter number of arc additions.
\end{corollary}

\begin{theorem} \label{theorem:closenesswharddirected}
    \textsc{Directed Closeness Improvement} is NP-hard and W[2]-hard with respect to the number~$k$ of edge additions on directed acyclic graphs.
\end{theorem}
\begin{proof}
    The proof uses a parameterized reduction from \textsc{Set Cover} with the parameter number of subsets~$k$, which is be W[2]-hard~\cite{downey2012parameterized}. Let~$I = (\mathcal{F} = \{F_1,\ldots,F_m\}, U = \{s_1,\ldots,s_n\}, k)$ be a \textsc{Set Cover} instance. We reduce~$I$ to a \textsc{Closeness Improvement} instance~$I' = (G = (V,A),z,k,k + \frac{n}{2})$, where~$G$ is a directed, unweighted graph constructed as follows: For each~$s_i \in U$ and each~$F_j \in \mathcal{F}$, we add a vertex~$v_i$ or~$u_j$ to the graph, respectively. Furthermore, if~$s_i \in F_j$ for any~$s_i \in U, F_j \in \mathcal{F}$, then we add an arc~$(v_j,u_i)$. Finally, we add a vertex~$z$ to the constructed graph. We provide an example in \cref{figure:directedcloseness}. It is easy to see that the constructed directed graph is acyclic.
    
    Before showing the correctness of the reduction, we state and prove the following observation: The closeness centrality of~$z$ can be maximally increased by adding arcs from~$z$ to~$v_j$. 
    
    First of all, the closeness centrality of~$z$ can be maximally increased if the source of each added arc is~$z$. Otherwise, if~$z$ is the target of an arc, then we either introduced a loop, or the source of the arc remains unreachable from~$z$. If~$z$ is neither the source nor the target of the arc, then all introduced shortest paths containing this arc become even shorter if we replace the source of the arc by~$z$. Finally, an arc~$(z,u_i)$ can be replaced by~$(z,v_j)$, where~$(v_j,u_i) \in E$. By adding the arc~$(z,u_i)$, the distance from~$z$ to~$u_i$ is decreased to 1. An arc~$(z,v_j)$ decreases the distance from~$z$ to~$v_j$ to 1 and the distance of at least one more vertex~$u_i$ to 2. Hence, by adding arcs from~$z$ to~$v_j$, we obtain a larger closeness centrality of~$z$ compared to adding edges from~$z$ to~$u_i$. 

    We show that the reduction is correct, that is,~$I$ is a \YI{} if and only if~$I'$ is a \YI{}. 
    
   ~$\Rightarrow$: If~$I$ is a \YI{}, then there is an~$\mathcal{F}' \subseteq \mathcal{F}$ of size~$k$ such that~$\bigcup_{F_j \in \mathcal{F}'} = U$. By adding~$k$ arcs~$(z,v_j)$,~$F_j \in \mathcal{F}'$, there are~$k$ vertices with distance 1 from~$z$, and each vertex in~$\{u_1,\ldots,u_n\}$ has distance~$2$ from~$z$. Hence,~$c_z$ can be increased to~$k + \frac{n}{2}$ and~$I'$ is a \YI{}.

   ~$\Leftarrow$: If~$I$ is not a \YI{}, then there is no such set~$\mathcal{F}' \subseteq \mathcal{F}$ of size~$k$ such that~$\cup_{F_j \in \mathcal{F}'} = U$. 
    
     After adding~$k$ arcs from~$z$ to vertices in~$\{v_1,\ldots,v_m\}$, there is at least one vertex~$u_i$ such that there is no path from~$z$ to~$u_i$. Summing up,~$c_z$ can be increased to at most~$k + \frac{n'}{2}$ for~$n' < n$ and~$I'$ is a \NI{}.        
\end{proof}

In the next theorem, we
slightly modify the reduction in the proof of
\cref{theorem:closenesswharddirected} in order to show that \textsc{Directed
Closeness Improvement} remains W[2]-hard on directed graphs with diameter 4.

\begin{theorem} \label{theorem:closenessdirecteddiameter}
    \textsc{Directed Closeness Improvement} is NP-hard and W[2]-hard with respect to the
    number $k$ of edge additions on directed graphs with diameter~4.
\end{theorem}
\begin{proof}
    Let~$I = (\mathcal{F} = \{F_1,\ldots,F_m\}, U = \{s_1,\ldots,s_n\}, k)$ be a \textsc{Set Cover} instance. We construct a \textsc{Directed Closeness Improvement} instance $I' = (G=(V,E),z,k,r)$ as follows. First we construct a directed graph as described in the reduction in the proof of \cref{theorem:closenesswharddirected}. Then we add $m$ vertices $w_i$ and $2m$ arcs $(z,w_i),(w_i,v_i)$ for each $1 \leq i \leq m$. Additionally, for each $u_i, v_i$ and $w_i \in V$, we add the arcs $(u_i,z), (v_i,z)$ and $(w_i,z)$ to $G$. Finally, we set $r=k + 2n - \tfrac{k}{2}$.

    The constructed graph is a directed graph with diameter 3: From $z$, the length of shortest paths to the other vertices is at most 3. The distance from each vertex $w_i, v_i$ and $u_i$ to any vertex $u_j$ is at most 4, and the distance from these vertices to any vertex $v_j$ is at most 3. Hence, $G$ is a strongly connected directed graph with diameter 4.
    
    Analogously to the reduction in \cref{theorem:closenesswharddirected}, there is an optimal solution for $I'$ which only contains arcs where $z$ is the source and some of the vertices $v_i$ are the target - the proof for this statement is the same as the one in the referred theorem.
    
    It remains to show that the reduction is correct, that is $I'$ is a \YI{} if and only if $I$ is a \YI{}:
    
    ~$\Rightarrow$: If~$I$ is a \YI{}, then there is an~$\mathcal{F}' \subseteq \mathcal{F}$ of size~$k$ such that~$\bigcup_{F_j \in \mathcal{F}'} = U$. By adding~$k$ arcs~$(z,v_j)$,~$F_j \in \mathcal{F}'$, there are~$k$ vertices $v_i$ with distance 1 from~$z$, and each vertex in~$\{u_1,\ldots,u_n\}$ has distance~$2$ from~$z$. Moreover, the other $n-k$ vertices $v_i$ have distance 2 from $z$, and each vertex $w_i$ has distance 1 from $z$. Hence,~$c_z$ can be increased to~$r=k + n + \tfrac{n + (n-k)}{2}$ and~$I'$ is a \YI{}.

    ~$\Leftarrow$: If~$I$ is not a \YI{}, then there is no such set~$\mathcal{F}' \subseteq \mathcal{F}$ of size~$k$ such that~$\cup_{F_j \in \mathcal{F}'} = U$. 
    
After adding~$k$ arcs from~$z$ to vertices in~$\{v_1,\ldots,v_m\}$, there is at least one vertex~$u_i$ such that there is no path from~$z$ to~$u_i$. Hence, there are $n$ vertices $w_i$ and $k$ vertices $v_i$ with distance 1 from $z$ and $n-k$ vertices $u_i$ with distance 2 from $z$. Furthermore, there are $n' < n$ vertices $u_i$ with distance 2 from $z$, and there is at least one vertex $u_i$ with distance 3 from $z$. Summing up,~$c_z$ can be increased to at most~$k + n + \frac{n' + (n-k)}{2} + \tfrac{1}{3}$ for~$n' < n$ and~$I'$ is a \NI{}.         
\end{proof}

The computational complexity of \textsc{Directed Closeness Centrality} on
directed graphs with diameter 3 remains open. However, analogously to the
problem variant with undirected input graphs, it is not hard to show that the problem is polynomial-time solvable on graphs with diameter at most 2.

\section{Betweenness Centrality} \label{chapter:betweenness}

\looseness=-1 We now investigate the problem of increasing the betweenness
centrality of a vertex in a graph by inserting a certain number of
edges into the graph. We remark that the betweenness centrality of a vertex in an
undirected graph can be computed in $O(n\cdot m)$ time~\cite{brandes2001faster}.
We show that, similar to \CI{}, \BI{} is W[2]-hard with respect to the
parameter number of edge additions and FPT with respect to the combination of the number of edge additions and the
distance to a cluster graph. 

First, we make an important observation that
will help us in our proofs. Analogous to \cref{lemma:closenessendpointz}, we
show that to improve the betweenness of a vertex by adding edges, it always
makes sense to add only edges adjacent to that vertex.

\begin{lemma} \label{lemma:betweennessendpoints}
    Let~$I = (G,z,k,r)$ be a \textsc{Betweenness Improvement} instance. If~$I$ is a \YI{}, then there is an optimal solution that only contains edges where one endpoint is~$z$.
\end{lemma}
\begin{proof}
    Let~$S$ be a solution for~$I$, and let~$e := \{u_i,u_j\} \in S$. Furthermore, assume that~$e$ introduces at least one shortest path containing~$z$ (if it does not, then~$e$ can be replaced by any edge containing~$z$). Without loss of generality, assume~$u_i$ precedes~$u_j$ on each of these paths. Then by replacing~$e$ by~$e' := \{z,u_j\}$ in~$S$, the distance between~$z$ and~$u_j$ decreases to 1 and the shortest paths previously containing~$e$ now contain~$e'$. Hence,~$b_z$ does not decrease.  
\end{proof}

Hence, if we compute a solution for some \BI{} instance, we need to find a
subset of the graph's vertices of size~$k$ such that adding an edge between~$z$
and these vertices maximally increases the betweenness centrality of~$z$. This
directly implies the following corollary:

\begin{corollary} \label{corollary:BIXP}
    \BI{} is solvable in $O(n^k \cdot (n + m))$ time where $k$ is the number of edge additions and thus is in XP with respect to the parameter number of edge additions.
\end{corollary} 

\paragraph{Hardness Results.}

We show that \BI{} is W[2]-hard with respect to the parameter number of edge additions by a parameterized reduction from \textsc{Dominating Set} on graphs with diameter 3.  Furthermore, we show that the problem is NP-hard on graphs with diameter 3 and H-index 4.

\begin{theorem} \label{theorem:betweennessundirected}
    \BI{} is NP-hard and W[2]-hard with respect to the parameter number~$k$ of edge additions on graphs with diameter~3. Moreover, unless the Exponential Time Hypothesis fails, \BI{} does not allow an algorithm with running time~$f(k)\cdot n^{o(k)}$.
\end{theorem}

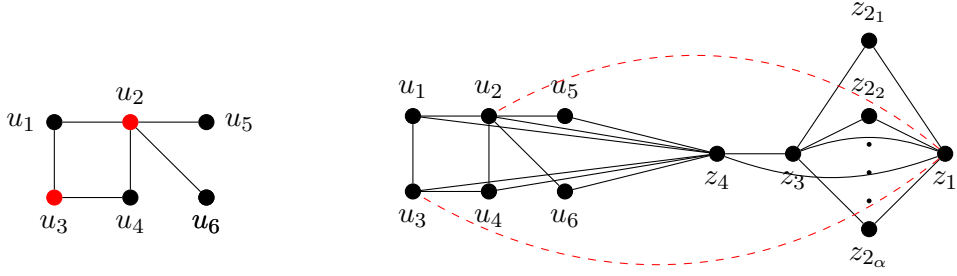
\begin{figure}[t!]
  \centering
  \hfill
  				\begin{tikzpicture}[scale=1]
				\tikzstyle{knoten}=[circle,draw,fill,minimum size=5pt,inner sep=2pt]
				\node[knoten,label={[label distance=0cm]180:$u_1$}] (K-1) at (0,0) {};
				\node[knoten,color=red,label={[label distance=0cm]90:$u_2$}] (K-2) at (1,0) {};
				\node[knoten,color=red,label={[label distance=0cm]270:$u_3$}] (K-3) at (0,-1) {};
				\node[knoten,label={[label distance=0cm]270:$u_4$}] (K-4) at (1,-1) {};
				\node[knoten,label={[label distance=0cm]0:$u_5$}] (K-5) at (2,0) {};
				\node[knoten,label={[label distance=0cm]270:$u_6$}] (K-6) at (2,-1) {};
				\node[knoten,label={[label distance=0cm]270:$u_6$}] (K-6) at (2,-1) {};
				
				\node[knoten,color=white, fill=white] (K-invisible) at (2,-2.25) {};
                
                \node[color=white,knoten] (K-8) at (0,1.9) {};
               % \node[knoten,label={[label distance=0cm]270:$z_{2_\alpha}$}] (K-10) at (0,-1.5) {};
				
				% Connect nodes
				\foreach \i / \j in {1/2,1/3,2/4,2/5,2/6, 3/4}{
					\path (K-\i) edge[-] (K-\j);
				}
                              \end{tikzpicture}
                              \hfill
				\begin{tikzpicture}[scale=1]
				\tikzstyle{knoten}=[circle,draw,fill,minimum size=5pt,inner sep=2pt]
				\tikzstyle{ph}=[circle,draw,fill,minimum size=1pt,inner sep=.5pt]
				\node[knoten,label={[label distance=0cm]90:$u_1$}] (K-1) at (0,0) {};
				\node[knoten,label={[label distance=0cm]90:$u_2$}] (K-2) at (1,0) {};
				\node[knoten,label={[label distance=0cm]270:$u_3$}] (K-3) at (0,-1) {};
				\node[knoten,label={[label distance=0cm]270:$u_4$}] (K-4) at (1,-1) {};
				\node[knoten,label={[label distance=0cm]90:$u_5$}] (K-5) at (2,0) {};
				\node[knoten,label={[label distance=0cm]270:$u_6$}] (K-6) at (2,-1) {};
				\node[knoten,label={[label distance=0cm]270:$z_1$}] (K-7) at (7,-.5) {};
				\node[knoten,label={[label distance=0cm]90:$z_{2_1}$}] (K-8) at (6,1) {};
				\node[knoten,label={[label distance=0cm]90:$z_{2_2}$}] (K-9) at (6,0) {};
				
				\node[ph] (ph1) at (6,-.375) {};
				\node[ph] (ph2) at (6,-.75) {};
				\node[ph] (ph3) at (6,-1.125) {};
				
				\node[knoten,label={[label distance=0cm]270:$z_{2_\alpha}$}] (K-10) at (6,-1.5) {};
				\node[knoten,label={[label distance=0cm]270:$z_3$}] (K-11) at (5,-.5) {};
				\node[knoten,label={[label distance=0cm]270:$z_4$}] (K-12) at (4,-.5) {};
				% Connect nodes
				\foreach \i / \j in {1/2,1/3,2/4,2/5,2/6, 3/4, 11/12}{
					\path (K-\i) edge[-] (K-\j);
				}
			
				\foreach \i in {1,2,3,4,5,6}{
					\path (K-\i) edge[-] (K-12);
				}
			
			    \path (K-7) edge[-, bend right=20](K-11);
				\path (K-7) edge[-, bend left=20](K-12);
					
				\foreach \i in {8,9,10}{
					\path (K-\i) edge[-] (K-7);
					\path (K-\i) edge[-] (K-11);
				}
			
				\path (K-7) edge[color=red,-, bend right=35, dashed](K-2);
				\path (K-7) edge[color=red,-, bend left=35, dashed](K-3);

                              \end{tikzpicture}
                              \hfill\mbox{}
	\caption{Parameterized reduction from \DS{} to \BI{}. Left: A \DS{} instance
	($I = (G,k=2)$). The red colored vertices $u_2, u_3$ form a solution. Right: The
	constructed \textsc{Betweenness Improvement} instance $I' = (G,z_1,k,r)$. The red, dashed edges form a solution.}% for $I'$.}
	\label{figure:betweenunweighted}
\end{figure}

\begin{proof} 
    We give a parameterized reduction from \DS{}, which also directly implies the running time lower bound when assuming ETH~\cite{cygan2015parameterized}. Let~$I = (G = (U,E), k)$ be a \DS{} instance, where~$U = \{u_1,\ldots,u_n\}$ . We construct a \BI{} instance~\[I' = \left(G' = (V,E'), z_1, k, r = \alpha k + \frac{2}{3} \alpha (n-k)  + \frac{1}{2}\left(k + \alpha + \binom{\alpha}{2}\right)\right),\] where~$\alpha > \frac{3k(k-1)}{2}$. The graph~$G'$ is constructed as follows. For each~$u_i \in U$, we add a vertex~$u'_i$ to~$G'$. Also, for each
    edge~$\{u_i,u_j\} \in E$, we add an edge~$\{u'_i,u'_j\}$ to~$E'$. We set~$U'
    := \{u'_1,\ldots,u'_n\}$. Next, we add the vertices~$\{z_1, z_3, z_4\}$
    and~$Z_2 = \{z_{2_1},\ldots,z_{2_\alpha}\}$ to~$G'$. For each~$z_{2_i} \in
    Z_2$, we add two edges~$\{z_1,z_{2_i}\}$ and~$\{z_{2_i},z_3\}$ to~$G'$.
    Furthermore, we add the edges~$\{z_1,z_3\}, \{z_1,z_4\}$ and~$\{z_3,z_4\}$.
    Finally, for each vertex~$u'_i \in U'$, we add an edge~$\{z_4,u'_i\}$.
    \cref{figure:betweenunweighted} illustrates the construction. It is easy to check that $G'$ has diameter 3.

    As~$z_1$ is adjacent to all vertices except the ones in~$U'$, a solution~$S$
    for~$I'$ contains only edges where one endpoint is~$z_1$ and each other one is in~$U'$ (\cref{lemma:betweennessendpoints}). 

We now show that~$I'$ is a \YI{} if and only if~$I$ is a \YI: First, if~$I$ is a \NI{}, we show that there is an upper bound~$r_u < r$ such that~$b_{z_1}$ can be increased to at most~$r_u$ by adding at most~$k$ edges to~$G'$. Second, if~$I$ is a \YI{}, we provide a lower bound~$r_\ell \geq r$ such that~$b_{z_1}$ can be increased to at least~$r_\ell$ by adding at most~$k$ edges to~$G'$. Both~$r_\ell$ and~$r_u$ depend on~$\alpha$, which determines the size of~$G'$. Finally, we determine a minimum value for~$\alpha$ such that~$r_\ell$ and~$r_u$ are strict bounds.

   ~$\Rightarrow$ The input graph contains a dominating set~$U_{DS} \subseteq U$ of size~$k$. We say that~$U'_{DS}$ is the set of vertices in the constructed graph which correspond to the vertices in~$U_{DS}$. Then, by adding~$k$ edges between~$z_1$ and the vertices in~$U'_{DS}$, for the following pairs of vertices there are shortest paths containing~$z_1$:
    \begin{compactitem}
        \item For each pair~$(u' \in U'_{DS}, z \in Z_2)$, there is one shortest path of length~2, containing~$z_1$. The number of such pairs is~$\alpha k$. 
        \item For each pair~$(u' \in U' \setminus U'_{DS}, z \in Z_2)$, two out of three shortest paths of length 3 between~$u'$ and the vertices in~$z$ contain~$z_1$: One contains~$z_1$ and a member of the dominating set, one contains~$z_1$ and~$z_4$, and one contains~$z_3$ and~$z_4$. The number of such pairs is~$\alpha (n-k)$. 
        \item For each pair~$(u' \in U'_{DS}, z_3)$, there are two shortest paths of length~2 between~$u'$ and~$z_3$: One contains~$z_1$, the other one contains~$z_4$. The number of such pairs is~$k$. 
        \item For each pair~$(z_{2_i}, z_{2_j} \in Z_2 \mid i \neq j)$, there are two shortest paths of length~2 between~$z_{2_i}$ and~$z_{2_j}$: One contains~$z_1$, the other one contains~$z_3$. The number of such pairs is~$\binom{\alpha}{2}$. 
        \item For each pair~$(z_{2_i} \in Z_2, z_4)$, there are two shortest paths of length~2: One contains~$z_3$ and the other one contains~$z_1$. The number of such pairs is~$\alpha$
    \end{compactitem}
  
    In total,\[b_{z_1} \geq \alpha k + \frac{2 \alpha (n-k)}{3}  + \frac{k}{2} + \frac{\binom{\alpha}{2}}{2}  + \frac{\alpha}{2},\]
    which can be simplified to \[b_{z_1} \geq \alpha k + \frac{2 \alpha (n-k)}{3}  + \frac{k + \alpha + \binom{\alpha}{2}}{2} =: r_\ell.\] 
    
   ~$\Leftarrow$ We prove the reverse direction by contraposition. That is, we show that~$I'$ is a \NI{} if~$I$ is a \NI{}. If the input instance does not admit a dominating set of size at most~$k$, then there is at least one vertex which cannot be dominated. We analyze the number of shortest paths in the constructed \BI{} instance after adding~$k$ edges between~$z_1$ and vertices in~$U'$. We set~$U'' \subseteq U' := \{u'_i \in U' \mid \{z_1,u'_i\} \in S\}$. Furthermore, let~$\ell$ be the number of vertices that are undominated in~$G'$ after adding the edges in~$S$, i.e.\ which are not adjacent to~$z_1$ and which do not have a neighbor adjacent to~$z_1$. As~$G$ does not admit a dominating set of size~$k$, it holds that~$\ell \geq 1$.
    \begin{compactitem}
    	\item For each pair~$(u' \in U'', z \in Z_2)$, there is one shortest path of length~2, containing~$z_1$. The number of such pairs is~$\alpha k$. 
    	\item For each pair~$(u' \in U' \setminus U'', z \in Z_2)$ where~$u'$ is a neighbor of one of the vertices in~$U''$, two out of three shortest paths of length 3 between~$u'$ and the vertices in~$z$ contain~$z_1$: One contains~$z_1$ and a member of the dominating set, one contains~$z_1$ and~$z_4$, and one contains~$z_3$ and~$z_4$. The number of such pairs is~$\alpha (n-k-\ell)$. 
    	\item For each pair~$(u'_i, u'_j \in U')$, there is a path of length~2 containing~$z_4$. If the vertices in~$U'$ are not adjacent, then this is the shortest path. Additionally, there may be another shortest path containing~$z_1$ of length~2, introduced by the edges in~$S$. Hence, for each of up to~$\binom{k}{2}$ pairs of vertices, one out of two shortest path contain~$z_1$.
    	
    	\item For each pair~$(u' \in U' \setminus U'', z \in Z_2)$ where~$u'$ is \textit{not} a neighbor of one of the vertices in~$U''$, there are two shortest paths betweens~$u'$ and~$z$ of length 3: One contains~$z_1$ and~$z_4$, the other one contains~$z_3$ and~$z_4$. The number of such pairs is~$\alpha \ell$. 
    	\item For each pair~$(u' \in U'', z_3)$, there are two shortest paths of length~2 between~$u'$ and~$z_3$: One contains~$z_1$, the other one contains~$z_4$. The number of such pairs is~$k$. 
    	\item For each pair~$(z_{2_i}, z_{2_j} \in Z_2 \mid i \neq j)$, there are two shortest paths of length~2 between~$z_{2_i}$ and~$z_{2_j}$: One contains~$z_1$, the other one contains~$z_3$. The number of such pairs is~$\binom{\alpha}{2}$. 
    	\item For each pair~$(z_{2_i} \in Z_2, z_4)$, there are two shortest paths of length~2: One contains~$z_3$ and the other one contains~$z_1$. The number of such pairs is~$\alpha$.
    \end{compactitem}
    
    In total, \[b_{z_1} \leq \alpha k + \frac{2 \alpha (n-k-\ell)}{3}  + \frac{\binom{k}{2}}{2} + \frac{\alpha \ell}{2} + \frac{k}{2} + \frac{\binom{\alpha}{2}}{2}  + \frac{\alpha}{2},\] which can be simplified to \[b_{z_1} \leq \alpha k + \frac{2 \alpha (n-k-\ell)}{3}  + \frac{\binom{k}{2} + \alpha (\ell + 1) + k + \binom{\alpha}{2}}{2} =: r_u.\]     
      
    In the last step, we need to determine a proper value for~$\alpha$ such that~$r_u < r_\ell$. Hence, the inequality that needs to be satisfied is   
    { \small
	    \begin{align*}
	    \alpha k + \frac{2 \alpha (n-k-\ell)}{3}  + \frac{\binom{k}{2} + \alpha (\ell + 1) + k + \binom{\alpha}{2}}{2} < \alpha k + \frac{2 \alpha (n-k)}{3}  + \frac{k + \alpha + \binom{\alpha}{2}}{2}
	    \end{align*}
	} %
     for each~$n,k,\ell \in \mathbb{N}, k \leq n, 1 \leq \ell \leq n$. This equation can be transformed to 
     
     \begin{align*}
     \frac{\alpha \ell}{3} > \binom{k}{2}.
     \end{align*}
     By setting~$\ell = 1$ and transforming the binomial coefficient, we get
     
     \begin{align*}
     \alpha > \frac{3k(k-1)}{2}.
     \end{align*}
     
     Hence, by setting~$\alpha$ to a value strictly larger than~$\frac{3k(k-1)}{2}$, the reduction is correct. Furthermore, the reduction is computable in fpt time: As the size of~$I'$ is polynomial to the size of~$I$,~$G'$ can be constructed even in polynomial time.       
\end{proof}

By closer inspection of the reduction, we can also show that \BI{} remains hard
on graphs with diameter 3 and H\nobreakdash-index 4.

\begin{corollary} \label{corollary:betweennesshindex}
     \BI{} is NP-hard on graphs with diameter 3 and H\nobreakdash-index 4. 
\end{corollary}
\begin{proof}Let~$I = (G,k)$ be a \DS{} instance, where~$G$ is a graph with maximum degree three. Let~$G'$ be the graph constructed by the reduction used in the proof of \cref{theorem:betweennessundirected}. Then each vertex except~$z_1,z_3,$ and $z_4$ has degree at most four. Hence, the h-index of~$G'$ is at most four. As \DS{} is NP-hard even on planar graphs with degree three \citep{Garey:1990:CIG:574848}, \BI{} remains NP-hard on graphs with h-index four. Furthermore, the constructed graph has diameter 3.  
\end{proof}

\paragraph{Algorithmic Result.} \label{section:positivebetweenness}
 
We also derive a positive result for \BI{}. We show that the problem is
fixed-parameter tractable with respect to the combined parameter distance to cluster and number of edge additions.

\begin{theorem} \label{theorem:distancetoclusterbetweenness}
   \BI{} is solvable in time~$2^{O(2^{2^{\ell}}\cdot k \log k)} \cdot n^{O(1)}$,
   where~$\ell$ is the distance of~$G$ to a cluster graph, and thus is in FPT
   with respect to the combined parameter~$(k, \ell)$.
\end{theorem}

\begin{proof}
   Let~$(G,z,k,r)$ be a \BI{} instance, where the set~$V_{\VDS} \subset V$ is a
   cluster vertex deletion set of size~$\ell$, that is, $G[V \setminus V_\VDS]$ is a cluster graph with connected components (\emph{clusters})~$\{C_1, \ldots, C_s\} =: C$. Since a cluster vertex deletion set of size~$\ell$ can be found in $O(1.92^\ell\cdot (n + m))$~time if it exists~\cite{boral_fast_2016,huffner2010fixed}, we may assume that $V_\VDS$ is given. The basic idea is similar to \cref{theorem:distancetoclustercloseness}. First, we determine the intersection of an optimal solution with $V_\VDS$. To find the vertices in $V \setminus V_\VDS$ we assign signatures to clusters and vertices in clusters based on their neighborhood in $V_\VDS$. We then find the signatures in an optimal solution and the optimal vertices for each signature. A difference to \cref{theorem:distancetoclustercloseness} is that, once we have determined the signatures of vertices in an optimal solution, it still matters how many vertices we take for each signature.

   Let $V^*$ be the set of endpoints different from~$z$ of the edges in an optimal solution. By \cref{lemma:betweennessendpoints} we may assume that $|V^*| = k$. The first step in the algorithm is to iterate over all $2^\ell$ possibilities for putting $V^* \cap V_\VDS$ in the output solution. Assume henceforth that we are in the iteration in which we have found $V^* \cap V_\VDS$.

   \looseness=-1   Define for each cluster~$C_i$ its \emph{cluster signature} as the set of neighbors of $C_i$ in $V_\VDS \cup \{z\}$. From $V^*$ we get a subset $S$ of the set of all $2^\ell$ possible cluster signatures by putting into $S$ all signatures of clusters which have nonempty intersection with~$V^*$. That is, $\S = \{N(C_i) \cap (V_\VDS \cup \{z\}) \mid i \in \{1, \ldots, s\}\}$, where $N(C_i) = \bigcup_{v \in C_i}N(v)$. The second step in the algorithm is to iterate over all $2^{2^\ell}$ possibilities for $\S$. Assume below that we are in the iteration in which we have found $\S$.

   Define for each vertex $v \in V \setminus (V_\VDS \cup \{z\})$ its \emph{vertex signature} as the set $N(v) \cap (V_\VDS \cup \{z\})$. From $V^*$, for each cluster signature $S \in \S$, we obtain a family $T_S$ of sets of vertex signatures by, for each cluster~$C_i$ with signature $S$ that has nonempty intersection with $V^*$, putting into $T_S$ the set $\{N(v) \cap V_\VDS\mid v \in C_i \cap V_\VDS\}$. The third step in the algorithm is to iterate for each $S \in \S$ over all $2^{2^\ell}$ possible families $T_S$. In total, these are at most $2^\ell\cdot 2^{2^{2^{\ell}}}$ possibilities. Assume henceforth that we are in the iteration in which we have found $T_S$ for each $S \in \S$.

   As fourth step in the algorithm we find for each $S \in \S$ and each $S' \in T_S$ the number $n_{S, S'}$ of clusters~$C_i$ such that $C_i \cap V^* \neq \emptyset$, $C_i$ has signature~$S'$, and the set of vertex signatures of vertices in $C_i \cap V^*$ is exactly~$S'$. We do this by iterating over all at most $(2^{\ell}\cdot 2^{2^{\ell}})^k$ possibilities. Assume henceforth that we are in the iteration in which we have found $n_{S, S'}$ for each $S \in \S$ and each $S' \in T_S$.

   \looseness=-1 As a fifth step in the algorithm we find for each of the $n_{S, S'}$ clusters~$C_i$ as above, for each vertex signature in~$s \in S'$ the number $n_{S, S', s}$ of vertices in $C_i \cap V^*$ with signature~$s$. Again, we iterate over all at most $(2^{\ell}\cdot 2^{2^\ell} \cdot 2^\ell \cdot k)^k$ possibilities. Assume henceforth that we are in the iteration in which we have found $n_{S, S', s}$.

\looseness=-1  Say that a cluster~$C_i$ is \emph{eligible} for $S, S'$ if it has cluster signature~$S$ and for each vertex signature~$s \in S'$ there are $n_{S, S', s}$ vertices with signature~$s$ in~$C_i$. We now claim that, without loss of generality, among clusters that are eligible for~$S, S'$, set $V^*$ contains only vertices from the $k$ largest such clusters. Assume otherwise. Hence, there is a cluster~$C_i$ among the $k$ largest clusters eligible for~$S, S'$ and a cluster~$C_j$ which is eligible for~$S, S'$ but not among the $k$ largest such clusters. (Recall that we are in an iteration in which we have found $S, S'$, $n_{S, S'}$, and $n_{S, S', s}$ as defined and hence, $C_j$ exists.) Obtain $W^*$ from $V^*$ by replacing each vertex in $C_j \cap V^*$ with a vertex in $C_i$ with the same signature; call the vertices in $C_i \cap W^*$ the \emph{replacements} of the vertices in~$V^*$. The betweenness centrality of~$z$ with respect to $W^*$ is at least the one % betweenness centrality of~$z$
with respect to~$V^*$. Indeed, each shortest path with respect to~$V^*$ that contains~$z$ and some vertices in~$C_j \cap V^*$ induces a shortest path with respect to~$W^*$ containing~$z$ and the corresponding replacements in~$C_i \cap W^*$. Thus, the claim holds.

   The sixth and final step in the algorithm is thus to try all possibilities to mark $n_{S, S'}$ clusters which are eligible for~$S, S'$ and to put, for each marked cluster and each $s \in S'$ a set of $n_{S, S', s}$ arbitrary vertices of signature $s$ in the marked cluster into the output solution. There are at most $(2^{\ell} \cdot 2^{2^{\ell}} \cdot k^2)^k$ possibilities. By the claim and since we can replace vertices with the same signatures in the marked clusters in $V^*$, in one of the tried possibilities, we will find an optimal solution. 
\end{proof}

\paragraph{Directed Betweenness Improvement.}
\label{section:betweennessdirected}

\looseness=-1 We now cover results for the problem of improving the betweenness
centrality of a vertex in a directed, unweighted graphs. First, we define
betweenness centrality for directed, unweighted graphs, as the definition due to
\citet{freeman1977set} only measures the centrality over all unordered subsets
of vertices of size two. A very natural definition, which is equivalent to the
one used in further literature (e.g.\ by \citet{WHITE1994335}) is to measure the ratio of shortest paths containing a certain vertex $z$ for both orders of any pair of vertices:
$b_z = \sum_{s \in V} \sum
            \frac{\sigma_{stz}}{\sigma_{st}}$. Herein $s, t \neq z$ and the second sum is taken over all $t \in V$ such that $t \neq s$ and $\sigma_{st} \neq 0$.
Using this definition, \textsc{Directed Betweenness Improvement} is defined analogously to \BI.

Analogously to the undirected problem variant, we show that we can maximize the betweenness centrality of a vertex~$z$ by adding arcs incident to~$z$. 

\begin{lemma}\label{lemma:betdirectedpaths}
 If a \DBI{} instance~$I = (G = (V,A),z,k,r)$ is a \YI{}, then there is a solution~$S$ that only contains arcs where either the source or the target is~$z$.
\end{lemma}
\begin{proof}
    Assume~$S$ contains an arc~$(u_1,u_2)$ such that~$u_1 \neq z$ and~$u_2 \neq z$. Let~$v_1,v_2 \in V$ such that~$(u_1,u_2)$ introduced a shortest path from~$v_1$ to~$v_2$ containing~$z$ and the arc~$(u_1,u_2)$. It is clear that~$u_2$ must have been connected to~$z$ by a path before adding the arc~$(u_1,u_2)$. Furthermore, it is clear~$\sigma_{v_1v_2z} \leq 1$, as there max be other paths from~$v_1$ to~$v_2$ not containing~$z$. 
    
    However, the shortest paths introduced by~$(u_1,u_2)$ necessarily contain~$u_1$ and~$u_2$; hence, these paths can be contracted by replacing~$(u_1,u_2)$ by~$(u_1,z)$. By this, we do not decrease~$b_z$: Let~$\ell$ be the length of a shortest path from~$v_1$ to~$v_2$ which contains~$(u_1,u_2)$ and~$z$. Then, after replacing~$(u_1, u_2)$ by~$(u_1,z)$, there is exactly one shortest path of length~$\ell' < \ell$ from~$v_1$ to~$v_2$. Hence,~$\sigma_{stz} = 1$.  
\end{proof}

However, note that a solution~$S$ for a \YI{}~$I = (G,z,k,r)$ may also contain arcs where~$z$ is the source. For instance,~$(G = \{z,v_1,v_2\}, A = \{(v_1,z)\}, z, 1, 1)$ is a \YI{}  with solution~$S = \{(z,v_2)\}$. 

\begin{corollary} \label{corollary:dbiXP}
    \DBI{} is solvable in $O((2n)^k\cdot (n+m))$ time where $k$ is
    the number of edge additions, and thus is in XP with respect to the parameter number of edge additions.
\end{corollary}

\noindent Substantial improvement of this running time is unlikely, as \cref{theorem:betweennesswharddirected} shows. 

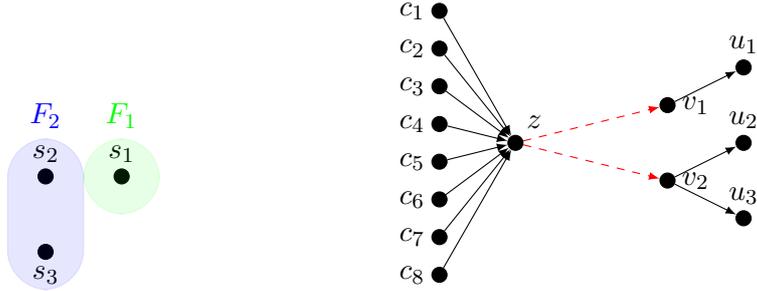
\begin{figure}[t!] 
    \centering
    \hfill
	\begin{tikzpicture}[scale=1,style=transform shape]
	\tikzstyle{knoten}=[circle,draw,fill,minimum size=5pt,inner sep=2pt,style=transform shape]
    %\tikzstyle{family}=[ellipse [x radius=1, y radius=2],draw,minimum size=10pt,inner sep=10pt,style=transform shape]
	\node[knoten,label={[label distance=-0.05cm]90:$s_1$}] (K-1) at (1,-1) {};
	\node[knoten,label={[label distance=-0.05cm]90:$s_2$}] (K-2) at (0,-1) {};
	\node[knoten,label={[label distance=-0.05cm]270:$s_3$}] (K-3) at (0,-2) {};

    %F2
    \draw[color=blue, fill=blue, opacity=0.1, rounded corners=15pt]
    (-0.5,-2.5) rectangle ++(1,2);;
    
    %F4
    \draw[color=green, fill=green, opacity=0.1, rounded corners=15pt]
    (0.5,-1.5) rectangle ++(1,1);

    \node[label={[color=blue,label distance=-0.3cm]270:$F_2$}, outer sep=2pt] at (0,0) {};
    
    \node[label={[color=green,label distance=-0.3cm]270:$F_1$}, outer sep=2pt] at (1,0) {};
    
	% Connect nodes
	%\foreach \i / \j in {1/2,1/3,2/4,2/5,2/6, 3/4}{
	%	\path (K-\i) edge[-] (K-\j);
	%}
	\end{tikzpicture}
    \hfill
\begin{tikzpicture}[scale=1,style=transform shape]
	\tikzstyle{knoten}=[circle,draw,fill,minimum size=5pt,inner sep=2pt,style=transform shape]
    
	\node[knoten,label={[label distance=-0.05cm]90:$u_1$}] (U-1) at (2,1) {};
    \node[knoten,label={[label distance=-0.05cm]90:$u_2$}] (U-2) at (2,0) {};
    \node[knoten,label={[label distance=-0.05cm]90:$u_3$}] (U-3) at (2,-1) {};
    
    \node[knoten,label={[label distance=-0.05cm]0:$v_1$}] (V-1) at (1,0.5) {};
    \node[knoten,label={[label distance=-0.05cm]0:$v_2$}] (V-2) at (1,-0.5) {};

    \node[knoten,label={[label distance=-0.05cm]78.75:$z$}] (Z) at (-1,0) {};
	
	%\node[knoten,color=white,label={[label distance=0cm]270:$$}] (K-55) at (2,-1.8) {};
	% Connect nodes

    %\path (Z2) edge[-latex] (Z);

    \foreach \i in {1,...,8}{
        \node[knoten,label={[label distance=-0.05cm]180:$c_\i$}] (Z\i) at (-2,2.25-\i/2) {};
        \path (Z\i) edge[-latex] (Z);
    }
    
	\foreach \i / \j in {1/1,2/2,3/2}{
		\path (V-\j) edge[-latex] (U-\i);
	}

    \path (Z) edge[-latex,color=red,dashed] (V-1);
	\path (Z) edge[-latex,color=red,dashed] (V-2);
    
    \end{tikzpicture}
    \hfill\mbox{}
    \caption{Parameterized reduction from \textsc{Set Cover} to \DBI{}. Left: A \textsc{Set Cover} instance~$I = (U,\mathcal{F},k=2$) with solution~$\{F_1,F_2\}$. Right: The constructed \DBI{} instance~$I' = (G,z,k=2,r=4\frac{5}{6})$. The red dashed edges imply a solution for~$I'$.}
    \label{figure:dbiarcadditions}
\end{figure}

\begin{theorem} \label{theorem:betweennesswharddirected}
    \DBI{} is NP-hard and W[2]-hard with respect to the parameter number of arc additions $k$ on directed acyclic graphs.
\end{theorem}
\begin{proof}
    We show a parameterized reduction from \textsc{Set Cover}. Let~$I = (\mathcal{F} = \{F_1,\ldots,F_m\}, U = \{s_1,\ldots,s_n\})$ be a \textsc{Set Cover} instance. We construct a \DBI{} instance~$I' = (G = (V,A),z,k,k(1+n) + n))$, where~$G$ is a directed, unweighted graph. The construction is as follows: 
    \begin{compactitem}
        \item For each~$s_i \in U$, add a vertex~$u_i$. Set~$V_U := \{u_1,\ldots,u_n\}$.
        \item For each~$F_j \in \mathcal{F}$, add a vertex~$v_j$. Set~$V_V := \{v_1,\ldots,v_m\}$.
        \item Add the vertex~$z$.
        \item Add the vertices~$c_1,\ldots,c_{m(m+n-1)}$; the set of these vertices is denoted as~$V_C$.
        \item For each~$c \in V_C$, add the arcs~$(c,z)$.
        \item For each~$u_i \in V_U$ and each~$v_j \in V_V$, add an arc~$(v_j,u_i)$ if~$s_i \in F_J$.
    \end{compactitem}
    In \cref{figure:dbiarcadditions}, the reduction is illustrated. It is easy to see that the constructed directed graph is acyclic.

    Let~$I$ be a \YI{} and~$S$ be an arc set of size at most~$k$, such that~$b_z \geq r$ in~$G' = (V,A \cup S)$. We now show that for each arc~$a \in S$, the source is~$z$ and the target is one of the vertices in~$V_V$. First, from \cref{lemma:betdirectedpaths} we know that there is a solution~$S'$ of the same size where~$z$ is an endpoint of each arc~$a \in S'$. Hence, in the following we assume that for each~$a \in S$, one of its endpoints is~$z$. 
    
   Moreover, if a solution~$S$ contains an arc~$(z,u_i), u_i \in V_U$, we can replace it by an arc~$(z,v_i), v_i \in V_V$ such that~$s_i \in F_j$ without decreasing~$b_z$: The arc~$(z,u_i)$ introduces paths from the vertices~$z$ and all its predecessors to~$u_i$.
   By replacing~$(z,u_i)$ by~$(z,v_j)$, the paths remain, but additionally paths from~$z$ and its predecessors to~$v_j$ are added. Hence,~$b_z$ does not decrease. 
   
   Furthermore, by adding the vertex set~$V_C$ of size~$m(m+n-1)$, we ensure that by adding arcs where the source is $z$, we obtain more (shortest) paths containing~$z$ than by adding arcs where the endpoint is $z$: Each arc from~$z$ to a vertex in~$V_V$ introduces at least~$m(m+n-1)$ shortest paths containing~$z$. However, adding an arc from a vertex in~$V_U$ to~$z$ introduces at most~$m ((m-1)+(n-1))$ paths containing~$z$: Each vertex in~$V_U$ has at most~$m$ predecessors in~$V_V$. Furthermore,~$z$ may have at most~$m$ successors in~$V_V$ and at most $(n-1)$ successors in $V_U$. Hence, by adding an arc from a vertex in $V_U$ to $z$, $c_z$ is increased by at most ~$m ((m-1)+(n-1))$. 
       
   We now show that the reduction is correct, i.e.\ that~$I$ is a \YI{} if and only if~$I'$ is a \YI{}.

   ~$\Rightarrow$: If~$I$ is a \YI{}, then there is a~$\mathcal{F}' \subseteq \mathcal{F}$ of size~$k$  such that~$\cup_{F_j \in \mathcal{F}'} = U$. By adding arcs~$(z,v_j)$ for each~$F_j \in \mathcal{F}'$, the following shortest paths contain~$z$:     
    \begin{compactitem} 
        \item For each~$v_j$ such that~$F_j \in \mathcal{F}'$ and each~$c \in V_C$, there is a shortest path from~$c$ to~$v_j$ containing~$z$. As~$|\mathcal{F}'| = k$, the number of such shortest paths is~$k(m(m+n-1))$.
        \item For each~$u_i$ and each~$c \in V_C$, there is a shortest path from~$c$ to~$u_i$ containing~$z$. In total, the number of such paths is~$n(m(m+n-1))$.
    \end{compactitem}
    
    Hence,~$b_z$ can be increased to~$(k+n)(m(m+n-1))$ and~$I'$ is a \YI{}.  
    
   ~$\Leftarrow$: If~$I$ is not a \YI{}, then there is no such set~$\mathcal{F}' \subseteq \mathcal{F}$ of size~$k$ such that~$\sum_{\mathcal{F}'} = U$. Let~$S$ be a set of size~$k$ which contains arcs from vertices~$v_j$ to~$z$. The graph~$G' = (V, A \cup S)$ contains the following shortest paths, each containing~$z$: 
    \begin{compactitem}
        \item For each~$v \in V_V$ which is the endpoint in an arc in~$S$, and each~$c \in V_C$, there is a shortest path from~$c$ to~$v$ containing~$z$. As the target of all arcs in~$S$ is a vertex in~$V_V$ and~$|S| = k$, the number of such shortest paths is~$k(m(m+n-1))$.
        \item As~$I$ is a \NI{}, there is at least one vertex~$u$ in~$G'$ such that there is no path from the vertices in~$V_C$ to~$u$. Hence, the number of  paths from vertices in~$V_C$ to vertices in~$V_U$ is at most~$n-1(m(m+n-1))$.
    \end{compactitem}
    Hence,~$b_z$ can be increased to at most~$(k+n-1)(m(m+n-1))$ and~$I'$ is a \YI{}.  
\end{proof}

\section{Conclusion and Outlook}
We studied the (parameterized) complexity of \CI{} and \BI\ with respect to the number~$k$ of added edges and (unfortunately) obtained mostly hardness results even in several special cases that are relevant to practice. On the plus side, we obtained tractability results relating to the vertex-deletion distance to cluster graphs.

\looseness=-1 Our tractability results yield running times that are impractical and need to be improved. Some further questions that we left open
are as follows. First, it is not hard to show that \CI\ polynomial-time solvable on graphs of
diameter~2. Is this also true for diameter~3? 
As we showed, for diameter~4 it is NP-hard. Noticeable is also that the problem seems to be harder on disconnected graphs. In particular, our reductions also imply NP-hardness for \emph{disconnected} graphs where every connected component has diameter~2.

There seem to be similarities between \DS\ and \CI, as
indicated by our hardness reductions. \DS\ is
fixed-parameter tractable with respect to the combined parameter
maximum degree and~$k$. Does the same hold for \CI?  Similar questions
extend to \BI. 
For \BI\ it would also be interesting to
see, whether in our fixed-parameter algorithm for the combined parameter solution size~$k$ and the distance to cluster graph, we can remove the dependency on~$k$.

\bibliographystyle{abbrvnat}
\bibliography{sources}

\end{document}